\documentclass[11pt]{article}
\usepackage[centertags]{amsmath}
\usepackage{amsthm,amsfonts,amssymb}
\title{A Root-Free Splitting-Lemma for\\ Systems of Linear Differential
Equations}
\author{Eckhard Pfl\"{u}gel\\
Faculty of Computing, Information Systems and Mathematics\\
Kingston University\\
Penrhyn Road\\
Kingston upon Thames\\
Surrey KT1 2EE\\
United Kingdom\\
E.Pfluegel@kingston.ac.uk}
\date{}
\newtheorem{theorem}{Theorem}[section]
\newtheorem{proposition}{Proposition}[section]
\newtheorem{lemma}{Lemma}[section]
\newtheorem{definition}{Definition}[section]
\renewenvironment{proof}
        {\noindent {\bf Proof}\hspace*{0.1cm}}{\hspace{\fill}$\Box$\bigskip\\}
\newcounter{remcounter}
\newenvironment{remark}
        {\noindent {\bf Remark\stepcounter{remcounter} \arabic{section}.\arabic{remcounter}.}\hspace*{0.05cm}}{}
\newenvironment{mat}
    {\left( \begin{array}}{\end{array} \right) }

\newcommand{\C}{{\mathbb C}}
\newcommand{\Z}{{\mathbb Z}}
\newcommand{\N}{{\mathbb N}}

\newcommand{\GLn}[1]{{\rm GL}(n, {#1})}

\newcommand{\diag}{{\rm diag}}
\newcommand{\wspec}{\omega\mbox{-}{\rm spec}}
\newcommand{\wlspec}[1]{\omega^{#1}\mbox{-}{\rm spec}}
\newcommand{\wpspec}{\omega^p\mbox{-}{\rm spec}}
\newcommand{\spec}{{\rm spec}}

\begin{document}

\bibliographystyle{plain}
\maketitle

\begin{abstract}
We consider the formal reduction of a system of linear
differential equations and show that, if the system can be
block-diagonalised through transformation with a ramified
Shearing-transformation and following application of the Splitting
Lemma \cite{Was67}, and if the spectra of the leading block
matrices of the ramified system satisfy a symmetry condition, this
block-diagonalisation can also be achieved through an unramified
transformation. Combined with classical results by Turritin
\cite{Tur55} and Wasow \cite{Was67} as well as work by Balser
\cite{Bal2000}, this yields a constructive and simple proof of the
existence of an unramified block-diagonal form from which formal
invariants such as the Newton polygon can be read directly. Our
result is particularly useful for designing efficient algorithms
for the formal reduction of the system.
\end{abstract}

\begin{center}
{\small {\bf Keywords:} Formal Reduction of Systems of Linear
Differential Equations,\\
Formal Solutions, Newton Polygons}
\end{center}

\section{Introduction}
When studying the formal reduction of a system of linear differential equations
\begin{equation}
    \label{system}
                x\frac{dy}{dx} = A(x)y
\end{equation}
where $y$ is a vector with $n \ge 2$ components and $A$ a square
formal meromorphic power series matrix  of dimension $n$ of the form
\[
          A(x) = x^{-r}\sum_{j=0}^\infty A_{j} x^j\quad(A_0 \neq 0)
\]
with pole order $r>0 $, the structure of the leading matrix $A_0 $
allows to reduce the problem to several problems of smaller size
whenever $A_0 $ has several eigenvalues. The well-known {\em
Splitting Lemma} \cite{Was67} states that if $A_0 $ is
block-diagonal
\[
    A_0 =
    \left(%
\begin{array}{cc}
  A^{11}_0 & 0 \\
  0 & A^{22}_0\\
\end{array}%
\right)
\]
with the additional condition that $A^{11}_0$ and $A^{22}_0$ have
no common eigenvalue, there exists a formal transformation matrix
\begin{equation}
    \label{Tx}
    T(x)= \sum_{j=0}^\infty T_{j} x^j\quad(T_0 = I)
\end{equation}
such that the change of variable $ y = Tz$ transforms the system
(\ref{system}) into a new system
\begin{equation}
    \label{newsystem}
                x\frac{dz}{dx} = B(x)z
\end{equation}
where
\[
B = \left(%
\begin{array}{cc}
  B^{11} & 0 \\
  0 & B^{22} \\
\end{array}%
\right)
\]
is of same pole order $r $ and block-diagonal with the same block
partition as in $A_0 $. The matrix $B $ is computed by
\begin{equation}
    \label{transformation}
           B  = T[A]:=T^{-1}AT - xT^{-1}\frac{dT}{dx}.
\end{equation}
Using the Splitting Lemma it is hence sufficient to study the case
where the leading matrix $A_0 $ in (\ref{system}) has only one
eigenvalue. Using an {\em exponential shift} of the form
$y=\exp(\lambda/x^r)z$ where $\lambda $ is the unique eigenvalue
of $A_0 $ one can (and we will throughout this paper) assume that
$A_0 $ is nilpotent.\\

Several methods for finding transformation matrices which again
lead to non-nilpotent leading matrices have been suggested
\cite{Bal2000, Bar97a, Che90, Pfl00, Tur55, Was67}. It can be
shown that this, combined with the Splitting Lemma, gives rise to
a recursive procedure which decomposes the initial system into new
systems for which one has either $n = 1 $ or $r = 0$. The
structure of the matrix $W $ of a formal fundamental matrix
solution of the system
\begin{equation}\label{FFMS}
    Y(x)= F(x)x^{\Lambda} e^{W(x)}
\end{equation}
can be determined uniquely through this method. Here $F $ is an
invertible formal meromorphic matrix power series in a fractional
power of $x $, $\Lambda$ is a constant complex matrix commuting
with $W$ and $W $ is a diagonal matrix containing polynomials in
the same fractional power of $x$ without constant
terms.\\

For the purposes of this paper, it is useful to distinguish
between the following two types of transformation matrices:
\begin{enumerate}
    \item Matrices containing formal meromorphic power series in the variable $x $, whose determinant
          is not the zero series.
          We will refer to this type of transformations as {\em root-free transformations}.
          Two systems linked as in (\ref{transformation}) by such a root-free
          transformation shall be called {\em meromorphically equivalent} or short {\em equivalent}.
    \item Matrices having coefficients which are formal meromorphic power series in a fractional power of $x $,
          whose determinant is not the zero series. We will call these transformations {\em ramified
          transformations}. If $T $ is a ramified transformation, the smallest integer $q $ such
          that $T(x^q) $ is root-free is called the {\em ramification index} of $T $.
          We shall also say that $T $ is a {\em $q $-meromorphic transformation} and takes a system
          into a {\em $q$-meromorphically equivalent} system.
\end{enumerate}
In this paper, we are interested in the situation where one cannot
find transformations of the first type in order to obtain a system
with a non-nilpotent leading matrix. In other words, the
introducing of a ramification is necessary. This can be stated in
terms of formal solutions by saying that the dominant (negative)
power of $x $ in the matrix $W $, or alternatively the biggest
slope of the
Newton polygon of the system, is a rational number \cite{Bar97a, HilWaz86}.\\

In this case, the methods in \cite{Bal2000, Che90, Tur55, Was67}
apply a series of root-free, ramified and {\em
Shearing-transformations} (in \cite{Bar97a, Pfl00} a different
strategy is employed). A Shearing-transformation is a
transformation of the form
\[
    S(x) =  \left(%
\begin{array}{cccc}
  x^{p_1/q} &  &  &  \\
            & x^{p_2/q} &  &  \\
   &  & \ddots &  \\
   &  &  & x^{p_n/q} \\
\end{array}%
\right)
\]
where $p_j\in\Z $ and $q\in\N $. In \cite{Bal2000} it is shown
that it is always possible to achieve this by using exponential
shifts and a transformation of the form
\begin{equation}
    \label{balser_transformation}
    T(x) = R(x)S(x)
\end{equation}
where $R$ is a root-free transformation having a finite number of
nonzero terms and $S$ is a ramified Shearing-transformation. The
transformed system
\begin{equation}
    \label{qsystem}
                x\frac{dy}{dx} = \hat{A}(x)y
\end{equation}
has a coefficient matrix of the form
\begin{equation}\label{qmatrix}
              \hat{A}(x) = x^{-r}\sum_{j=p}^\infty \hat{A}_{j}
              x^{j/q}\quad(\hat{A}_p \neq 0, \;p\ge 0)
\end{equation}
where $p$ is relatively prime to $q$ and $\hat{A}_p $ has several
eigenvalues. Applying the Splitting Lemma to (\ref{qsystem}) then
yields a $q $-meromorphic transformation taking the system into a
new system whose coefficient matrix is block-diagonal. Hence the
remaining computations are carried out on matrices
containing ramified power series.\\

One may ask under which conditions there exists also a root-free
transformation which achieves a block-diagonalisation of the
original system (\ref{system}) {\em without} introducing
ramifications, and how to compute such a transformation.\\

We shall give a sufficient condition for the existence of such a
root-free transformation and also provide a constructive method
for computing its coefficients. Denote by $\spec(A) $ the set of
eigenvalues of a complex square matrix $A $. We will prove the
following theorem:
\begin{theorem}[\lq\lq Root-Free Splitting Lemma"]
\label{main_theorem} Consider the system (\ref{system}) and assume
there exists a Shearing-transformation $S$ of ramification index
$q$ taking the system into one of the form (\ref{qsystem}) such
that its leading matrix $\hat{A}_p$ is similar to a block-diagonal
matrix
\[
\hat{B}_p = \left(%
\begin{array}{cc}
  \hat{B}^{11}_p & 0 \\
  0 & \hat{B}^{22}_p \\
\end{array}%
\right)
\]
and suppose that for all $\lambda_1 \in \spec(\hat{B}_p^{11}) $
and $\lambda_2 \in \spec(\hat{B}_p^{22}) $ it holds $\lambda_1
\neq e^{2\pi ik/q}\lambda_2 $ for all $k\in \N $. Then there
exists a root-free transformation $H $ with the following
properties:
\begin{enumerate}
    \item $H $ transforms the system (\ref{system}) into an
        equivalent system with block-diagonal coefficient matrix
\[
{B} = \left(%
\begin{array}{cc}
  {B}^{11} & 0 \\
  0 & {B}^{22} \\
\end{array}%
\right)
\]
    where the block sizes match those in the matrix $\hat{B}_p$.
    \item The matrix $S[B]$ has the leading matrix $\hat{B}_p$ and
    the same pole order as $S[A]$.
    \item A finite number of coefficients of the root-free transformation $H $
can be computed from a finite number of the coefficients of the
system (\ref{system}).
\end{enumerate}
\end{theorem} The classical Splitting Lemma can be seen as a
particular case of this theorem by putting $S$ as the identity
matrix and $q=1$.\\

This paper is organised as following: in Section
\ref{splitting_review}, we review the classical Splitting Lemma.
In the following section we introduce a special class of systems
and give a variant of the Splitting Lemma, particular to this
class. Using this, we will give the proof of Theorem
\ref{main_theorem} in Section \ref{root_free} and illustrate the
benefits of our theorem concerning the formal
reduction in Section \ref{application} on an example.\\

\noindent {\bf Notations}: Throughout the paper, empty entries in
matrices are supposed to be filled with 0. We write
$\diag(a_1,\ldots,a_n)$ for a (block-)diagonal matrix whose
diagonal entries are the $a_i$. The valuation of a polynomial or
formal power series (with possibly negative or fractional
exponents) is the smallest occurring power in the variable $x$.
Other definitions of notations are made as they appear in the
text.

\section{Review of the Splitting Lemma}
\label{splitting_review}
 As we have previously mentioned, the
Splitting Lemma is a well-known result. Its proof is carried out
in a constructive fashion and gives a method for computing the
coefficients $T_j$ of the transformation matrix $T $ as in
(\ref{Tx}), see for example \cite{Bal2000, Bar97a, Was67}. We
repeat it here for reason of completeness. Also, we will formulate
it for $q $-meromorphic systems in preparation of the proof of
Lemma \ref{commutative_splitting}.
\begin{lemma}
\label{splitting} Consider the system (\ref{qsystem}) and assume
that  $\hat{A}_p$ is block-diagonal
\[
\hat{A}_p = \left(%
\begin{array}{cc}
  \hat{A}_p^{11} & 0 \\
  0 & \hat{A}_p^{22} \\
\end{array}%
\right)
\]
such that
\[
\spec(\hat{A}_p^{11})\cap\spec(\hat{A}_p^{22})=\emptyset.
\]
Then there exists a formal $q $-meromorphic transformation of the
form
\[
\hat{T}(x) = \sum_{j = 0}^{\infty}\hat{T}_j x^{j/q}\quad
(\hat{T}_0 = I)
\]
such that the transformed system is block-diagonal with the same
block partition as in $\hat{A}_p$.
\end{lemma}

\begin{proof}
We use a transformation of the special form
\[
\hat{T}(x) =\left(%
\begin{array}{cc}
  I & \hat{U}(x) \\
  \hat{V}(x) & I\\
\end{array}%
\right)
\]
with $\hat{U_0} = \hat{V_0}= 0 $. Denote by $\hat{B}$ the matrix
$\hat{T}[\hat{A}]$. Inserting the series expansion for $\hat{A},
\hat{B} $ and $\hat{T} $  and comparing coefficients gives the
recursion formula
\begin{equation}\label{recursion}
    \hat{A}_p\hat{T}_h-\hat{T}_h \hat{A}_p =
    \sum_{j = 1}^{h} (\hat{T}_{h-j}\hat{B}_{j+p}-\hat{A}_{j+p} \hat{T}_{h-j})+((p+h)/q-r) \hat{T}_{p+h-qr},\quad h> 0
\end{equation}
where $\hat{T}_{j} = 0 $ for $j<0 $. Equation (\ref{recursion}) is
of the form
\begin{equation}\label{recursion1}
      \hat{A}_p\hat{T}_h-\hat{T}_h \hat{A}_p = \hat{B}_{h+p}-\hat{A}_{h+ p}+\hat{R}_h
\end{equation}
where
\[
    \hat{R}_h = \sum_{j = 1}^{h-1} (\hat{T}_{h-j}\hat{B}_{j+p}-\hat{A}_{j+p} \hat{T}_{h-j})
                +((p+h)/q-r) \hat{T}_{p+h-qr}
\]
depends only on $\hat{B}_j $ with $j < h+ p $ and $\hat{T}_{j}$
with $j<h$. Using the special form of
\[
        \hat{T}_h = \begin{mat}{cc} 0 & \hat{U}_h\\ \hat{V}_h & 0 \end{mat}, \quad
        \hat{B}_h = \begin{mat}{cc} \hat{B}^{11}_h & 0\\ 0 & \hat{B}^{22}_h \end{mat}
\]
and decomposing $\hat{R}_h$ into block-structure accordingly gives
the following system of equations:
\begin{eqnarray}
      \hat{B}^{11}_{p +h}+\hat{R}^{11}_h &=& 0,\label{eq1}\\
      \hat{B}^{22}_{p +h}+\hat{R}^{22}_h &=& 0\label{eq2}
\end{eqnarray}
where $\hat{B}^{11}_{p +h}$ and $\hat{B}^{22}_{p +h}$ are unknown,
and
\begin{eqnarray}
      \hat{A}_p^{11}\hat{U}_h - \hat{U}_h \hat{A}_p^{22} &=& \hat{R}^{12}_h,\label{eq3}\\
      \hat{A}_p^{11}\hat{V}_h - \hat{V}_h \hat{A}_p^{22} &=&
      \hat{R}^{21}_h\nonumber
\end{eqnarray}
with unknowns $\hat{U}_h $ and $\hat{V}_h $. Given $\hat{R}_h$,
the first two equations (\ref{eq1}) and (\ref{eq2}) can be solved
by setting $\hat{B}^{11}_{p +h}=-\hat{R}^{11}_h$ and
$\hat{B}^{22}_{p +h}=-\hat{R}^{22}_h$. The remaining equations can
be solved uniquely for $\hat{U}_h $ and $\hat{V}_h$ because the
matrices $\hat{A}_p^{11}$ and $\hat{A}_p^{22}$ have no eigenvalues
in common, see e.g. \cite{gan59}.
\end{proof}

\section{On \boldmath $(\omega, P)$-Commutative Systems \unboldmath}

In this section, we study a particular class of $q $-meromorphic
systems. Starting point of our considerations was \cite[Lemma 5,
Section 3.3]{Bal2000} observing that a system transformed by a
Shearing-transformation has a special structure. We will state
this more generally and give conditions under which this special
structure is preserved by the Splitting-Lemma.\\
\begin{definition}
\label{commutative} Let $q> 1 $ be a positive integer, $\omega=
e^{2\pi i/q}$ and $P \in \C^{n\times n}$. We call a formal $q
$-meromorphic matrix $\hat{A}$ as in (\ref{qmatrix}) {\em
$(\omega, P)$-commutative} if
\[
    \hat{A}_j P =\omega^{j} P \hat{A}_j\quad (j \ge p).
\]
A system of the form (\ref{qsystem}) is called $(\omega, P)$-commutative
 if its coefficient matrix is $(\omega, P)$-commutative.
\end{definition}

\begin{remark}
The considerations in \cite{Bal2000} correspond, in our notation,
to the case of a $(\omega, P) $-commutative system where $P $ is
an invertible diagonal matrix. However, this restriction is not
necessary in this section and we will develop our theory
first for arbitrary matrices $P $.\\
\end{remark}

\begin{remark}
Two complex matrices $A$ and $B $ satisfying $A B =\omega B A $
are called {\em $\omega $-commutative} in \cite{Hol03}. In terms
of their notation, the $j$th coefficient of a $(\omega,
P)$-commutative matrix and the
matrix $P $ are $\omega^{j}$-commutative.\\
\end{remark}

The following lemma gives an alternative characterisation for
$(\omega, P)$-commutative systems which will be useful later. Note
that similar concepts are used in \cite{Bal2000}.
\begin{lemma}
\label{alternative_criterion}
\renewcommand{\labelenumi}{\roman{enumi})}
Consider a system of the form (\ref{qsystem}). Then the following
two statements are equivalent:
\begin{enumerate}
    \item $\hat{A}$ is $(\omega, P)$-commutative.
    \item $\hat{A}(x)P= P\hat{A}(e^{2\pi i}x) $.
\end{enumerate}
\end{lemma}
\begin{proof}
A direct calculation shows:
\begin{eqnarray*}
    \hat{A}_j P & = & \omega^{j} P \hat{A}_j \quad \forall j\ge p \\
    \Longleftrightarrow x^{-r} \sum_{j = p}^\infty \hat{A}_j x^{j/q} P & =
    & x^{-r}\sum_{j = p}^\infty e^{2\pi ij/q}P \hat{A}_j x^{j/q}\\
    \Longleftrightarrow \hat{A}(x)P & = & P\hat{A}(e^{2\pi i} x).
\end{eqnarray*}
\end{proof}
It is also straightforward to see that we have
\begin{lemma}
\label{unramified} Consider a system of the form (\ref{qsystem})
and suppose $\hat{A}$ is $(\omega, I)$-commutative where $I$
denotes the $n\times n$ identity matrix. Then $\hat{A}$ is an
unramified formal meromorphic power series matrix.
\end{lemma}
We make the following definition: for two eigenvalues $\lambda_1 $
and $\lambda_2 $ of $\hat{A}_p $ we define an equivalence relation
$\sim_l$ by
\[
    \lambda_1 \sim_l \lambda_2 \quad \Longleftrightarrow \quad \exists
    k\in\{0,\ldots, q-1\}: \lambda_1=\omega^{l k} \lambda_2
\]
and denote by $\wlspec{l}(A_p)$ the set
\[
    \{[\lambda]_{\sim_l}|\lambda \in \spec(\hat{A}_p)\}
\]
where we will, slightly abusing notation, identify $\lambda$ with
$[\lambda]_{\sim_l}$.\\

Given $C\in\GLn{\C} $, it is clear that the matrix
$C^{-1}\hat{A}C$ is $(\omega, C^{-1}PC)$-commutative.
\begin{lemma}
\label{prepare_system}
Let $\hat{A}$ as in (\ref{qmatrix}) be
$(\omega, P)$-commutative and let $\lambda_1 $ and  $\lambda_2 $
be two eigenvalues of $\hat{A}_p$ with $\lambda_1
\not\sim_p\lambda_2$. Then there exists $C\in\GLn{\C} $ such that
$\hat{B} = C^{-1}\hat{A}C $ is $(\omega, \tilde{P})$-commutative
and
\[
\hat{B}_p = \left(%
\begin{array}{cc}
  \hat{B}_p^{11} & 0 \\
  0 & \hat{B}_p^{22} \\
\end{array}%
\right),\quad \tilde{P} =
\left(%
\begin{array}{cc}
  \tilde{P}^{11} & 0 \\
  0 & \tilde{P}^{22}\\
\end{array}%
\right)
\]
with $\tilde{P} = C^{-1}P C $, $\lambda_1 \in
\spec(\hat{B}_p^{11}) $, $ \lambda_2 \in \spec(\hat{B}_p^{22}) $
and
$\wpspec(\hat{B}_p^{11})\cap\wpspec(\hat{B}_p^{22})=\emptyset$.
\end{lemma}
\begin{proof}
The existence of a matrix $C $ so that $\hat{B}_p =
C^{-1}\hat{A}_pC $ satisfies the conditions of the lemma can be
seen easily from elementary properties of matrix decomposition. We
will use techniques similar as in \cite{Hol03} in order to show
that $\tilde{P}$ has the required block diagonal structure. Let
\[
 \tilde{P}=
\left(%
\begin{array}{cc}
  \tilde{P}^{11} & \tilde{P}^{12} \\
  \tilde{P}^{21} & \tilde{P}^{22}\\
\end{array}%
\right)
\]
where the block partition matches that in the matrix $\hat{B}_p$.
Inserting into the equation $\hat{B}_p\tilde{P}=\omega ^
p\tilde{P}\hat{B}_p$ yields in particular the two conditions
\[
    \hat{B}^{11}_p\tilde{P}^{12}-\omega
    ^p\tilde{P}^{12}\hat{B}^{22}_p= 0
\]
and
\[
    \hat{B}^{22}_p\tilde{P}^{21}-\omega
    ^p\tilde{P}^{21}\hat{B}^{11}_p= 0.
\]
The first of these two equations is of the form
\[
    \hat{B}^{11}_pX-X\omega^p\hat{B}^{22}_p= 0.
\]
The assumption
$\wpspec(\hat{B}_p^{11})\cap\wpspec(\hat{B}_p^{22})=\emptyset$
implies that the matrices $\hat{B}_p^{11}$ and
$\omega^p\hat{B}_p^{22}$ have no eigenvalue in common. The above
equation therefore has the unique solution $\tilde{P}^{12} = 0 $.
A very similar argument applies to the second equation. This
proves the lemma.
\end{proof}
\begin{remark}
The matrices $C$, $\tilde{B}$ and $\tilde{P}$ are in general
not uniquely determined.\\
\end{remark}

We now show that for a $(\omega, P) $-commutative system which has
block-diagonal structure as in Lemma \ref{prepare_system},
application of the Splitting Lemma preserves the property of being
$(\omega, P) $-commutative.
\begin{lemma}[{\bf \lq\lq Splitting Lemma for \boldmath $ (\omega, P)$-Commutative \unboldmath Systems"}]
\label{commutative_splitting} Consider the system (\ref{qsystem})
and assume that $\hat{A}$ is $(\omega, P)$-commutative with
$\hat{A}_p$ and $P$ block-diagonal with blocks of same dimension
\[
\hat{A}_p = \left(%
\begin{array}{cc}
  \hat{A}_p^{11} & 0 \\
  0 & \hat{A}_p^{22} \\
\end{array}%
\right),\quad P =
\left(
\begin{array}{cc}
  P^{11} & 0 \\
  0 & P^{22}\\
\end{array}%
\right)
\]
such that
\[
\wpspec(\hat{A}_p^{11})\cap\wpspec(\hat{A}_p^{22})=\emptyset.
\]
Then there exists a $(\omega, P)$-commutative $q $-meromorphic
transformation of the form
\begin{equation}
 \hat{T}(x) =
\sum_{j=0}^{\infty}T_jx^{j/q}\quad (T_0 = I)
\end{equation}
such that the transformed system is $(\omega, P)$-commutative and
block-diagonal with the same block partition as in $\hat{A}_p $
and $P $.
\end{lemma}
\begin{proof}
The existence of the transformation $\hat{T} $ is given by Lemma
\ref{splitting}, the classical Splitting Lemma. What remains to
show is that $\hat{T} $ and the transformed system are $(\omega,
P)$-commutative. Denote by $\hat{B} $ the coefficient matrix of
the transformed system. Using the notations as in
(\ref{recursion}) and (\ref{recursion1}), we will show that the
following relations hold:
\begin{eqnarray}
     \hat{R}_k P &=& \omega^{p+k} P \hat{R}_k, \label{R}\\
     \hat{T}_k P &=& \omega^{k} P \hat{T}_k, \label{T}\\
     \hat{B}_{p +k} P &=& \omega^{p +k} P \hat{B}_{p +k}\label{B}
\end{eqnarray}
for $k\in\N $. The case $k = 0 $ holds trivially by putting
$\hat{R}_0 = 0 $ since $\hat{T}_0 = I$ and $ \hat{B}_p=
\hat{A}_p$. Let $h$  be an arbitrary positive integer. We will see
that if the above relations hold for
$k = 0,\ldots, h-1 $ then they hold for $h $. The claim follows then by induction.\\

We compute
\begin{eqnarray*}
     \hat{R}_h P  &=& \sum_{j = 1}^{h-1} (\hat{T}_{h-j}\hat{B}_{j+p}-\hat{A}_{j+p}\hat{T}_{h-j})P
                +((p+h)/q-r) \hat{T}_{p+h-qr}P\\
                &=&  \sum_{j = 1}^{h-1}\omega^{p+h} P (\hat{T}_{h-j}\hat{B}_{j+p}-\hat{A}_{j+p}\hat{T}_{h-j})
                +((p+h)/q-r)\omega^{p+h} P \hat{T}_{p+h-qr}\\
                &=& \omega^{p+h} P \hat{R}_h
\end{eqnarray*}
where we have used (\ref{T}) and  (\ref{B}) for $k = 0,\ldots, h-1
$ and the assumption that $\hat{A}$ is $(\omega, P)$-commutative.
This proves (\ref{R}) for $k =h$.\\

We decompose $\hat{R}_h $ into blocks accordingly to the block
structure of $\hat{B}_h $ and $P $ and find using (\ref{eq1})
\[
     \hat{B}_{p +k}^{11}P^{11}=-\hat{R}^{11}_k P^{11}=-\omega^{p +k} P^{11}
     \hat{R}^{11}_k=\omega^{p +k}\hat{B}_{p +k}^{11}P^{11}.
\]
We can show an analogous relationship for $\hat{B}_{p +k}^{22}$
and $P^{22}$ using (\ref{eq2}). Hence we can see that (\ref{B})
holds for $k
=h$.\\

It remains to show (\ref{T}), which is equivalent to showing
\begin{eqnarray}
     \hat{U}_h P^{22} &=& \omega^{h} P^{11} \hat{U}_h, \label{U}\\
     \hat{V}_h P^{11} &=& \omega^{h} P^{22} \hat{V}_h. \label{V}
\end{eqnarray}
We will only show that the first of these two equations holds, the
second can be dealt with similarly. Multiplying (\ref{eq3}) with
$\omega^{p} P^{11} $ on the left and with $P^{22}$ on the right
and combining the two equations yields
\[
          \hat{A}_p^{11}(\hat{U}_h P^{22} - \omega^h P^{11}\hat{U}_h) -
          (\hat{U}_h P^{22} - \omega^h P^{11}\hat{U}_h)\omega^{p}\hat{A}_p^{22} = 0.
\]
This equation is of the form
\[
          \hat{A}_p^{11}X - X \omega^{p}\hat{A}_p^{22} = 0.
\]
The assumption
$\wpspec(\hat{A}_p^{11})\cap\wpspec(\hat{A}_p^{22})=\emptyset$
implies that the matrices $\hat{A}_p^{11}$ and
$\omega^{p}\hat{A}_p^{22}$ have no eigenvalue in common. The above
equation therefore has the unique solution $X = 0 $, from  which
we conclude (\ref{U}). This completes the proof of the lemma.
\end{proof}
\begin{remark}
\label{block_commutative_remark}
We observe that the two block
matrices in the transformed system are
$(\omega,P^{11})$-commutative and $(\omega,P^{22})$-commutative
respectively.
\end{remark}

\section{A Root-Free Splitting Lemma}
\label{root_free} We define a {\em generalised
Shearing-transformation} as a transformation of the form $SC$
where $S$ is a Shearing-transformation and $C\in\GLn{\C}$.
\begin{proposition}
\label{commutative_criterion} Consider a system as in
(\ref{qsystem}) with leading matrix $\hat{A}_p$ and let $q\geq 2$.
The following statements are equivalent:
\renewcommand{\labelenumi}{\roman{enumi})}
\begin{enumerate}
    \item There exists a system as in (\ref{system}) and a generalised
    Shearing-transformation $\tilde{S}$ of ramifications index $q$
    such that $\tilde{S}[A]=\hat{A}$.
    \item The system (\ref{qsystem}) is $(\omega, \tilde{P})$-commutative, the matrix $\tilde{P}$ is similar to a diagonal matrix and
    $\spec(\tilde{P})\subseteq\{1,\omega,\omega^2,\ldots,\omega^{q-1}\} $.
    Furthermore, if $\lambda$ is an eigenvalue of $\hat{A}_{p}$ with
    multiplicity $s$, the numbers
    $\omega\lambda,\ldots,\omega^{(q-1)}\lambda$ are all eigenvalues of the same
    multiplicity $s$.
\end{enumerate}
\end{proposition}
\begin{proof}
We proof $i)\Rightarrow ii) $: let $\tilde{S}=SC$ be the
generalised Shearing-transformation. Since $\tilde{S}(e^{2\pi
i}x)=\tilde{S}(x)\tilde{P}$ where $\tilde{P}=C^{-1}PC$ with
$P=S(e^{2\pi i})$, we find with $\hat{A}=\tilde{S}[A]$
\begin{eqnarray*}
    \tilde{P}\hat{A}(e^{2\pi i}x) & = & \tilde{P}\tilde{S}^{-1}(e^{2\pi i}x)
A(e^{2\pi i}x)\tilde{S}(e^{2\pi i}x)-
        \tilde{P}x\tilde{S}^{-1}(e^{2\pi i}x)\tilde{S}'(e^{2\pi i}x)\\
        & = &\hat{A}(x)\tilde{P}\\
\end{eqnarray*}
showing that $\hat{A}$ is $(\omega, \tilde{P})$-commutative where
$\tilde{P}$ satisfies the stated properties. The claimed symmetry
in the spectrum of $\hat{A}_p$ can be shown as in the proofs of
\cite[Lemma 5, Section 3.3]{Bal2000} and \cite[Theorem 5]{Hol03}
since we have $\tilde{P}^{-1}\hat{A}_p\tilde{P}=\omega^p\hat{A}_p$
and $\omega^p$ is a primitive $q$th root of unity.\\

In order to prove the converse direction,
we first assume that $\tilde{P} =
\diag(\omega^{\alpha_1},\ldots,\omega^{\alpha_n})$ with
$\alpha_j\in\{0,\ldots, q-1\} $ and define the
Shearing-transformation
\[
            S(x) = \diag(  x^{-\alpha_1/q}, x^{-\alpha_2/q},
            \ldots, x^{-\alpha_n/q}).
\]
We observe that $S(e^{2\pi i}x)=\tilde{P}^{-1}S(x) $. Transform
the given system using this transformation $S$ and denote the
coefficient matrix of the transformed system by $B $. We compute
\begin{eqnarray*}
        B(e^{2\pi i}x)&=&S^{-1}(e^{2\pi i}x) \hat{A}(e^{2\pi i}x)S(e^{2\pi i}x)-
        xS^{-1}(e^{2\pi i}x)S'(e^{2\pi i}x)\\
                      &=&S^{-1}(x)\tilde{P}\hat{A}(e^{2\pi i}x)\tilde{P}^{-1}S(x)-xS^{-1}(x)S'(x)\\
                      & = & B(x)
\end{eqnarray*}
showing that $B$ is $(\omega, I)$-commutative and hence (Lemma
\ref{unramified}) $B $ must be a unramified formal meromorphic
power series matrix. The case of the general matrix $\tilde{P} $
follows by first applying a constant similarity transformation
which diagonalises $P $.
\end{proof}
We can now give the proof of our main theorem.\\

\begin{proof}{\bf of Theorem \ref{main_theorem}}
Let $C\in\GLn{\C}$ such that $\hat{B}=C^{-1}\hat{A}C$ has a
leading matrix $\hat{B}_p$ as in the assumptions of the Theorem.
In a similar way as in the proof of Lemma \ref{prepare_system}, we
can see that $C$ is block-diagonal, matching the block structure
of $\hat{B}_p$. Using this and Proposition
\ref{commutative_criterion}, we obtain that $\hat{B}$ is
$(\omega,\tilde{P})$-commutative where $\tilde{P}=C^{-1}S(e^{2\pi
i}x)C$ is similarly block-diagonal. Note that $p$ and $q$ being
relatively prime, the two conditions
$\wpspec(\hat{B}_p^{11})\cap\wpspec(\hat{B}_p^{22})=\emptyset$ and
$\wspec(\hat{B}_p^{11})\cap\wspec(\hat{B}_p^{22})=\emptyset$ are
equivalent. We can therefore apply Lemma
\ref{commutative_splitting} to $\hat{B}$ in order to obtain a
$(\omega,\tilde{P})$-commutative transformation matrix $\hat{T}$
such that $\hat{T}[\hat{B}]$ is $(\omega,\tilde{P})$-commutative
and block-diagonal with matching block structure.

We claim that the transformation matrix
\[
        H(x) = S(x)C\hat{T}(x)C^{-1}S^{-1}(x)
\]
is root-free satisfying the desired properties of the theorem. It
is clear that $H[A]$ is block-diagonal. But one verifies that $H$
is $(\omega,I)$-commutative and hence is root-free. The remaining
properties follow immediately.
\end{proof}

\section{Application for the Formal Reduction}
\label{application} Consider the situation where the system
(\ref{system}) is $q $-meromorphically equivalent to a system as
in (\ref{qsystem}) whose leading matrix $\hat{A}_p $ has several
eigenvalues but is not invertible. The exponential matrix
polynomial $W $ in a formal fundamental matrix solution
(\ref{FFMS}) is then
\[
    W(x) = \diag(w_1(x), w_2(x), \ldots, w_n(x))
\]
where the leading terms of diagonal entries of the form
\[
     w_k(x) = \lambda_k
     x^{-r+\frac{p}{q}}+\cdots\quad(\lambda_k\neq 0)
\]
are given by nonzero eigenvalues $\lambda_k$ of $\hat{A}_p $. The
diagonal entries having valuation greater than $-r+\frac{p}{q}$
correspond to eigenvalues zero. In particular, these entries might
involve ramifications different to $q$ or no ramifications at all.
Algorithms using the classical Splitting Lemma will not be able to
compute these entries without first introducing the ramification $q$.\\

In order to see how we can remedy this situation, we use the fact
that $\hat{A}_p $ is similar to a matrix of the form
\[
                \left(\begin{array}{cc}
                  \hat{B}_p^{11} & 0 \\
                  0 & \hat{B}_p^{22}\\
                \end{array}\right)
\]
where $\hat{B}_p^{11}$ is invertible and $\hat{B}_p^{22}$ is
nilpotent. Hence the conditions for Theorem \ref{main_theorem} are
satisfied and we will obtain a root-free transformation $H$ which
splits the system into
\begin{equation}
    \label{root_free_system}
                x\frac{dy}{dx} =
                \left(\begin{array}{cc}
                  B^{11} & 0 \\
                  0 & B^{22}\\
                \end{array}\right)y.
\end{equation}
This makes it possible to work independently on the two matrices
${B}^{11}$ and ${B}^{22}$: for the first matrix we can use a
Shearing-transformation introducing the (necessary) ramification
$q $. For the second matrix however we now recursively apply the
formal reduction process.\\

In order to illustrate this approach, consider the following
example with $n = 5 $, $r = 2 $ 
and
\[
    x^{-1}A(x) = \left( \begin {array}{ccccc} 0&{x}^{-3}&-{x}^{-1}&1&2\,{x}^{-1}\\\noalign{\medskip}-{x}^{-2}&{x}^{-1}&0&-{x}^{-1}&0
\\\noalign{\medskip}{x}^{-1}&1&0&{x}^{-3}&1\\\noalign{\medskip}1&-{x}^
{-1}&1&{x}^{-1}&{x}^{-3}\\\noalign{\medskip}{x}^{-1}&0&-3\,{x}^{-1}&0&
-1\end {array} \right).
\]
The Shearing-transformation $S(x) =\diag(S^{11}(x), S^{22}(x))$
with
\[
    S^{11}(x)=
                \left(\begin{array}{cc}
                  1 &  \\
                   & \sqrt{x}\\
                \end{array}\right),\quad
    S^{22}(x)) =
                \left(\begin{array}{ccc}
                  1 & & \\
                  & \sqrt{x}&\\
                    & &x\\
                \end{array}\right)
\]
transforms the system into a system of the form (\ref{qsystem})
with ramification index $q = 2 $, $p=1$ and block-diagonal leading
matrix with the two blocks
\[
    \hat{B}_1^{11}=
                \left(\begin{array}{cc}
                  0 & 1 \\
                 -1  & 0\\
                \end{array}\right),\quad
    \hat{B}_1^{22}=
                \left(\begin{array}{ccc}
                  0 &1 &0 \\
                  0& 0&1\\
                  0 & 0&0\\
                \end{array}\right)
\]
and the condition of Theorem \ref{main_theorem} is satisfied since
the first block matrix is invertible and the second is nilpotent.
We obtain a root-free transformation which is of the form
\[
    H(x) =
                    \left(\begin{array}{cc}
                  I &  U(x)\\
                  V(x) & I\\
                \end{array}\right),\quad
\]
with (we only have computed the first couple of terms)
\[
   U(x)=   \left( \begin {array}{ccc}
-6\,{x}^{2}+49\,{x}^{3}&-2\,x+18\,{x}^{2}&6\,x-48\,{x}^{2}\\
{x}^{2}-18\,{x}^{3}&-6\,{x}^{2}+46\,{x}^{3}&-2\,x+16\,{x}^{2}\\
                \end{array}\right)
\]
and
\[
    V(x)=   \left( \begin {array}{cc}
3\,x-29\,{x}^{2}+290\,{x}^{3}&1-9\,x+89\,{x}^{2}\\
-x+8\,{x}^{2}-88\,{x}^{3}&3\,x-28\,{x}^{2}+274\,{x}^{3}\\
-3\,{x}^{2}+28\,{x}^{3}-280\,{x}^{4}&-x+9\,{x}^{2}-88\,{x}^{3}
                \end{array}
                \right).
\]
Applying this transformation to the original system yields the
following block-diagonal matrix:
\[
    B(x) = \left( \begin {array}{ccccc} 0&{x}^{-3}-{x}^{-1}&&&\\\noalign{\medskip}
    -{x}^{-2}&{x}^{-1}&&&\\\noalign{\medskip}&&0&
{x}^{-3}&0\\\noalign{\medskip}&&0&{x}^{-1}&{x}^{-3}
\\\noalign{\medskip}&&-3\,{x}^{-1}&0&0\end {array} \right)+O(1).
\]
Applying the Shearing-transformation $S^{11}$ to the first block
matrix will result in a ramified system of smaller size and
invertible leading matrix equalling the first block of
$\hat{B}_1$. The formal reduction can now be applied to the second
block. In this example, it is found that another
Shearing-transformation of ramification index $q = 3$ results in a
system with invertible leading matrix. This decomposition can be
interpreted  as separation of the different slopes of the
Newton-polygon, see Theorem
\ref{newton_form} below.\\

We conclude that if the algorithm employed for computing the
transformation (\ref{transformation}) keeps the introduced
ramification minimal (as for example the algorithm in
\cite{Bal2000}), using the root-free Splitting Lemma allows the
recursive computation of $W$ using minimal ramifications. This
approach leads to a root-free transformation taking a system of
the form (\ref{system}) into a system from which all the leading
terms of the matrix $W $ (or, alternatively
the Newton polygon) can be determined directly. We state this as\\

\begin{theorem}
\label{newton_form} The given system (\ref{system}) is
meromorphically equivalent to a system $ x\frac{dy}{dx} = B(x)y$
where the matrix $B$ is block-diagonal
\[
    B(x)= \diag(B^{(1)}(x),B^{(2)}(x),\ldots,B^{(\nu)}(x))
\]
with $\nu \le n$ and there exist a diagonal transformation with
blocks of same sizes
\[
    S(x)= \diag(S^{(1)}(x),S^{(2)}(x),\ldots,S^{(\nu)}(x))
\]
with $S^{(k)}$ Shearing-transformations of ramification index
$q_k$ ($k=1,\ldots,\nu $) such that each of the matrices
$\hat{B}^{(k)} = S^{(k)}[B^{(k)}], k=1,\ldots,\nu $ has either no
pole at $x = 0 $ or an invertible leading matrix. In the latter
case, let $r_k=r-\frac{p_k}{q_k}$ $(\gcd(p_k,q_k)=1)$ be the pole
order and $\hat{B}_{p_k}^{(k)}$ the leading matrix of
$\hat{B}^{(k)}$. Then $|\wspec(\hat{B}_{p_k}^{(k)})|=1$ and if
$\lambda_k$ is an eigenvalue of $\hat{B}_{p_k}^{(k)}$ with
multiplicity $s_k$, the eigenvalues
$\omega\lambda_k,\ldots,\omega^{q_k-1}\lambda_k$ are all of the
same multiplicity $s_k$. There are $q_k s_k$ diagonal entries in
the matrix $W$ of the form
\[
     w_{k,j}(x) = \omega^{j}\lambda_k x^{-r_k}+\cdots\quad(j=
0,\ldots,q_k-1)
\]
where the dots denote terms with higher powers of $x $. The Newton
polygon of the system corresponding to this block admits a single
slope $r_k$ of length $q_ks_k$.
\end{theorem}
\begin{proof}
The proof follows from the procedure we have outlined in this
section, Proposition \ref{commutative_criterion}, and additional
application of the root-free Splitting Lemma to $B^{22}$ if
$B_p^{22}$ is similar to a block-diagonal matrix such that, for
each block, all eigenvalues are congruent modulo $\sim_p$.
\end{proof}
The issue of how to implement the root-free Splitting Lemma in a
computer algebra system deserves additional attention. In
particular, the question arises whether there exists a more direct
way of obtaining the root-free transformation matrix. Furthermore,
it seems likely that a combination of the results of this paper
together with our method in \cite{Pfl00}, where we have given a
different generalisation of the Splitting Lemma, will lead to a
significant improvement of the algorithmic formal reduction
of systems of linear differential equations.\\

\bibliography{ep}

\end{document}